\documentclass[a4paper,twoside,openany]{memoir}

\usepackage{amsmath}
\usepackage{amssymb}
\usepackage{amsthm}
\usepackage[hidelinks]{hyperref}
\usepackage[semibold,scale=0.875]{sourcecodepro}
\usepackage{listings}
\usepackage{adjustbox}
\usepackage{breakcites}

\chapterstyle{dash}
\renewcommand\chapnumfont{\huge\bfseries}
\copypagestyle{chapter}{plain}
\makeevenfoot{chapter}{}{\textbf{\thepage}}{}
\makeoddfoot{chapter}{}{\textbf{\thepage}}{}
\createmark{chapter}{left}{shownumber}{Chapter }{\hspace{4mm}}
\createmark{section}{right}{shownumber}{}{\hspace{4mm}}
\makepagestyle{mypage}
\makeheadrule{mypage}{\textwidth}{\normalrulethickness}
\makeevenhead{mypage}{\llap{\textbf{\thepage}\hspace{3\marginparsep}}\leftmark}{}{}
\makeoddhead{mypage}{}{}{\rightmark\rlap{\hspace{3\marginparsep}\textbf{\thepage}}}
\setsecnumdepth{subsection}
\maxtocdepth{subsection}
\newcommand{\thetitlepage}{{%
  \clearpage
  \thispagestyle{empty}
  \centering
  \Large\textsc{Bachelor Thesis}
  \vspace{0.4 cm}
  \hrule height 0.7 pt
  \vspace{0.3 cm}
  \huge\textbf{\textsc{\thetitle}}
  \vspace{0.3 cm}
  \hrule height 0.7 pt

  \vspace{0.4 cm}
  \LARGE\textsc{\theauthor}

  \vfill
  \includegraphics[width=0.4\textwidth]{\thelogo}
  \vfill

  \Large{\textbf{Supervisors} \\ \smallskip
  \thesupervisors

  \vspace{0.8 cm}
  \textbf{\theuniversity} \\ \smallskip
  \thegroup

  \vspace{0.8 cm}
  \thedate}
  \clearpage
  \thispagestyle{empty}
  \let\thefootnote\relax\footnotetext{
The original version of this thesis is available at \url{https://arxiv.org/abs/2310.19806v1}.
This version includes corrections of a few minor errors.
I would like to thank Vladimir Lifschitz for pointing out these errors.
  }
  \clearpage
}}

\newcommand{\university}[1]{\def\theuniversity{#1}}
\newcommand{\group}[1]{\def\thegroup{#1}}
\newcommand{\supervisors}[2]{\def\thesupervisors{#1 \\ #2}}
\newcommand{\logo}[1]{\def\thelogo{#1}}
\newcommand{\place}[1]{\def\theplace{#1}}

\newtheorem{theorem}{Theorem}
\newtheorem{lemma}{Lemma}
\theoremstyle{definition}
\newtheorem{definition}{Definition}[chapter]
\DeclareMathOperator{\mynot}{\mathit{not}}
\DeclareMathOperator{\myinf}{\mathit{inf}}
\DeclareMathOperator{\mysup}{\mathit{sup}}
\DeclareMathOperator{\valop}{\mathit{val}}
\newcommand{\val}[2]{\valop_{#1}(#2)}
\newcommand{\twodots}{\mathinner {\ldotp \ldotp}}
\newcommand{\atoms}{\mathcal{P}}
\newcommand{\paxioms}{\mathcal{A}}
\newcommand{\htint}[2]{\langle {#1},{#2} \rangle}

\lstdefinelanguage{input}{
	morekeywords={not},
	sensitive,
}
\lstdefinelanguage{human-readable}{
	morekeywords={and, or, not, exists, forall, true, false},
	sensitive,
}
\lstdefinelanguage{tptp}{
	otherkeywords={!, ?, \$true, \$false},
	morekeywords={},
	sensitive,
}
\lstdefinelanguage{anthem-output}{
	otherkeywords={info:},
	sensitive,
}
\lstnewenvironment{tabularlstlisting}[1][]{%
  \lstset{aboveskip=-1.3ex,belowskip=-2.1ex,#1}%
}{}

\newlength\excol
\newlength\extop
\newcommand{\anthem}{\textsc{anthem}}
\newcommand{\gringo}{\textsc{gringo}}
\newcommand{\clingo}{\textsc{clingo}}
\newcommand{\cvc}{\textsc{cvc4}}
\newcommand{\princess}{\textsc{princess}}
\newcommand{\vampire}{\textsc{vampire}}
\newcommand{\vampirep}{$\textsc{vampire}_\text{P}$}
\newcommand{\zipperposition}{\textsc{zipperposition}}

\title{Automated Verification of Equivalence Properties in Advanced Logic Programs}
\author{Jan Heuer}
\date{October 2020}
\place{Kleinmachnow}
\university{University of Potsdam}
\group{Knowledge Processing and \\ Information Systems}
\supervisors{Patrick Lühne, M. Sc.}{Prof. Dr. Torsten Schaub}
\logo{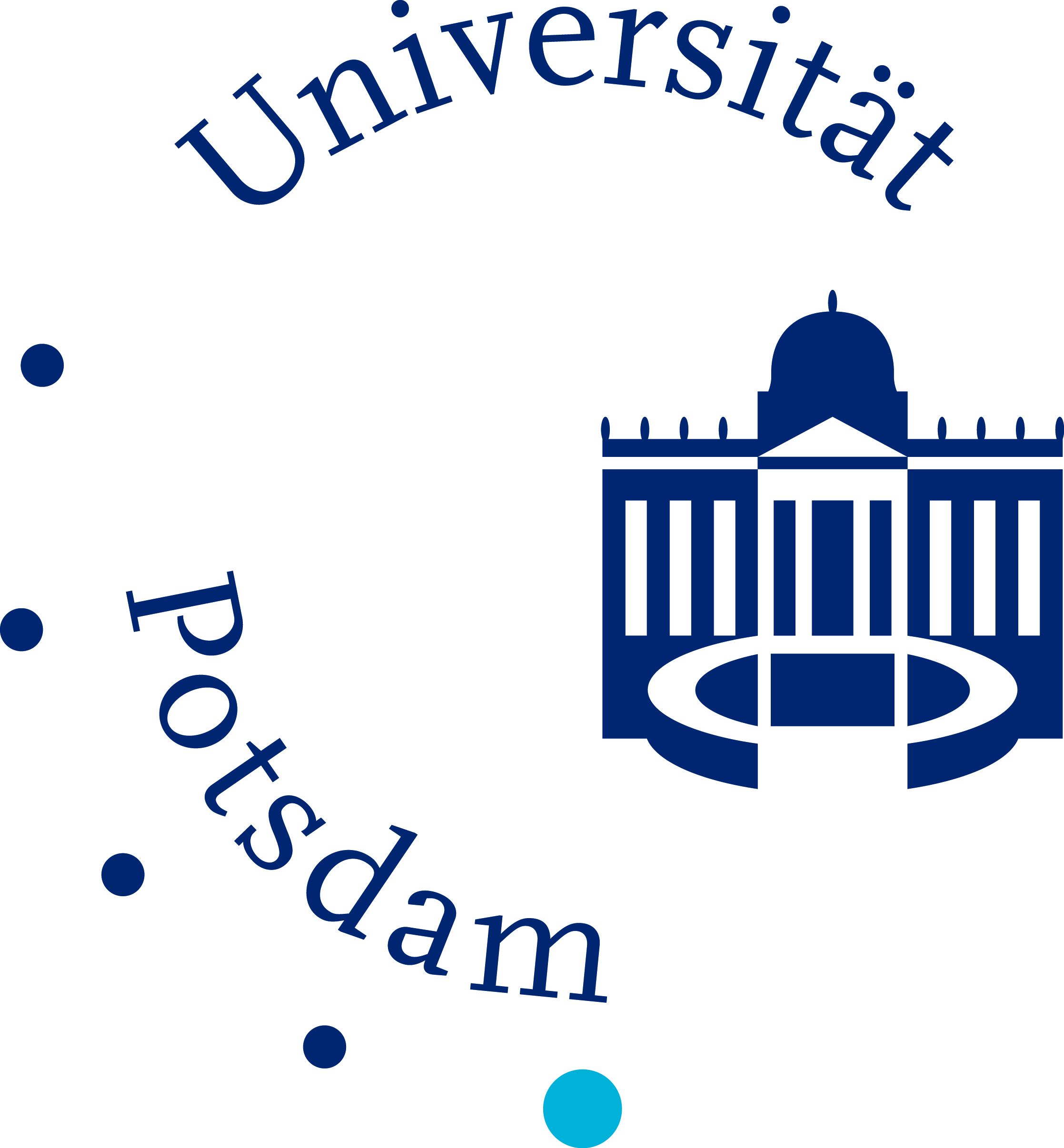}

\begin{document}
	\frontmatter
	\thetitlepage
	\tableofcontents*
	\chapter{Abstract}
  With the increase in industrial applications using Answer Set Programming, the need for formal verification tools, particularly for critical applications, has also increased.
  During the program optimisation process, it would be desirable to have a tool which can automatically verify whether an optimised subprogram can replace the original subprogram.
  Formally this corresponds to the problem of verifying the \emph{strong equivalence} of two programs.
  In order to do so, the translation tool \anthem{} was developed.
  It can be used in conjunction with an \emph{automated theorem prover} for \emph{classical logic} to verify that two programs are strongly equivalent.
  With the current version of \anthem{}, only the strong equivalence of \emph{positive programs} with a restricted input language can be verified.
  This is a result of the translation~$\tau^*$ implemented in \anthem{} that produces formulas in the \emph{logic of here-and-there}, which coincides with classical logic only for positive programs.
  This thesis extends \anthem{} in order to overcome these limitations.
  First, the transformation~$\sigma^*$ is presented, which transforms formulas from the logic of here-and-there to classical logic.
  A theorem formalises how~$\sigma^*$ can be used to express equivalence in the logic of here-and-there in classical logic.
  Second, the translation~$\tau^*$ is extended to programs containing \emph{pools}.
  Another theorem shows how~$\sigma^*$ can be combined with~$\tau^*$ to express the strong equivalence of two programs in classical logic.
  With~$\sigma^*$ and the extended~$\tau^*$, it is possible to express the strong equivalence of logic programs containing \emph{negation}, \emph{simple choices}, and \emph{pools}.
  Both the extended~$\tau^*$ and~$\sigma^*$ are implemented in a new version of \anthem{}.
  Several examples of logic programs containing pools, negation, and simple choice rules, which the new version of \anthem{} can translate to classical logic, are presented.
  Some available theorem provers (\cvc{}, \princess{}, \vampire{}, and \zipperposition{}) are compared for their ability to verify the strong equivalence of the different programs.
  \cvc{} and \vampire{} emerged as the most viable options to use in combination with \anthem{}.

	\chapter{Zusammenfassung}
  Mit dem Anstieg industrieller Anwendungen von Antwortmengenprogrammier\-ung hat sich auch der Bedarf von Werkzeugen zur formalen Verifikation erhöht, insbesondere für kritische Anwendungen.
  Während des Prozesses der Programmoptimierung wäre es wünschenswert ein Werkzeug zu haben, welches automatisch verifizieren kann, dass ein optimiertes Subprogramm das ursprüngliche Subprogramm ersetzen kann.
  Formal entspricht dies dem Problem die \emph{starke Äquivalenz} von zwei Programmen zu verifizieren.
  Um dies zu tun, wurde das Übersetzungstool \anthem{} entwickelt.
  Es kann mit \emph{automatisierten Theorembeweisern} für \emph{klassische Logik} kombiniert werden um zu verifizieren, dass zwei Programme stark äquivalent sind.
  Mit der aktuellen Version von \anthem{} ist es aber nur möglich die starke Äquivalenz von \emph{positiven Programmen} mit einer eingeschränkten Eingabesprache zu verifizieren.
  Das ist eine Folge der Übersetzung~$\tau^*$, welche in \anthem{} implementiert ist, die Formeln in der \emph{Logik von hier-und-da} produziert, welche nur für positive Programme mit der klassischen Logik übereinstimmt.
  Diese Arbeit erweitert \anthem{} um diese Einschränkungen zu überwinden.
  Zuerst wird die Transformation~$\sigma^*$ präsentiert, welche Formeln von der Logik von hier-und-da zu klassischer Logik transformiert.
  Ein Theorem formalisiert, wie~$\sigma^*$ benutzt werden kann um Äquivalenz in der Logik von hier-und-da in klassischer Logik auszudrücken.
  Zweitens wird die Übersetzung~$\tau^*$ für Programme mit \emph{Pools} erweitert.
  Ein weiteres Theorem zeigt, wie~$\sigma^*$ mit~$\tau^*$ kombiniert werden kann um die starke Äquivalenz von zwei Programmen in klassischer Logik auszudrücken.
  Mit~$\sigma^*$ und dem erweiterten~$\tau^*$ ist es möglich die starke Äquivalenz von Programmen mit \emph{Negation}, \emph{einfachen Choices} und \emph{Pools} auszudrücken.
  Sowohl das erweiterte~$\tau^*$ als auch~$\sigma^*$ sind in einer neuen Version von \anthem{} implementiert.
  Mehrere Beispiele von Logikprogrammen die Negation, einfache Choices sowie Pools enthalten, welche die neue Version von \anthem{} in klassische Logik übersetzen kann, werden präsentiert.
  Einige Theorembeweiser (\cvc{}, \princess{}, \vampire{} und \zipperposition{}) werden auf ihre Fähigkeit die starke Äquivalenz von den verschiedenen Programmen zu verifizieren verglichen.
  \cvc{} und \vampire{} stellen sich als die besten Optionen für die Nutzung in Kombination mit \anthem{} heraus.

	\mainmatter
  \chapter{Introduction}
\label{chap:introduction}
  In recent years, many industrial applications have been built using Answer Set Programming~\cite{ErdemEtAl2016ApplicationsAnswerSet,FalknerEtAl2018IndustrialApplicationsAnswer}.
  Because of this, the need for formal verification tools for Answer Set Programs has increased, as it is of great importance to know whether a developed program is actually correct.

  One possible application for formal verification is to decide whether an improved version of a program can replace the original without changing the behaviour on the application domain.
  This is an important question, as the development of an application usually starts with a simple but correct encoding, which is later on optimised.
  While it may be enough for some applications to run some tests on the new program, it is far better and for critical applications necessary to formally verify the correctness of the new program.
  This task is made even more complex as, in other programming paradigms, it is customary to split an application into smaller program parts.
  In order to decide if the optimised version of the program can replace the original one, one has to verify whether the new program works together correctly with the other program parts on any input.
  This idea is captured by the concept of \emph{strong equivalence} between two logic programs~\cite{LifschitzEtAl2001StronglyEquivalentLogic}, which is the question whether a program can replace another one in any context without modifying the semantics.

  In order to verify this question, the translation tool \anthem{} was developed~\cite{LifschitzEtAl2019VerifyingStrongEquivalence}.
  Using \anthem{} in conjunction with an \emph{automated theorem prover} makes it possible to automate the verification of strong equivalence of two logic programs.
  The translation of \anthem{} is based on the translation~$\tau$ used to define the semantics of the input language of \gringo{}~\cite{GebserEtAl2015AbstractGringo}.
  However, $\tau$~transforms logic programs into infinitary propositional formulas, which are not suitable for theorem proving.
  Therefore, \anthem{} instead uses a new translation~$\tau^*$, which transform programs into finite first-order formulas.
  But the output produced by \anthem{} has the semantics of the \emph{logic of here-and-there} whereas automated theorem provers work with the semantics of \emph{classical logic}.
  In the special case that the input programs are \emph{positive programs} the two semantics coincide, which makes the usage of \anthem{} possible.

  To illustrate the restrictions of the current version of \anthem{}, let us look at some example programs.
  The first program (Listing~\ref{lst:colour_atoms}) is part of the input to a graph colouring problem defining the available colours.

  \begin{lstlisting}[language=input, caption={Defining available colours in a graph colouring problem}, label={lst:colour_atoms}]
colour(r).
colour(g).
colour(b).
  \end{lstlisting}

  Such a program can be rewritten using a \emph{pool} as in Listing~\ref{lst:colour_pool} to make it more compact.

  \begin{lstlisting}[language=input, caption={Defining \lstinline{colour} using a pool}, label={lst:colour_pool}]
colour(r;g;b).
  \end{lstlisting}

  With the current version of \anthem{}, it is not possible to verify that this simplification is strongly equivalent, as \anthem{} terminates with an error stating that pools are not yet supported.

  Another example is the program in Listing~\ref{lst:transitive_rule}, which defines a transitive relation~$p$ for the set of values~$X$ for which~$q(X)$ holds.
  This program first uses a choice rule to define the relation~$p$ and then uses a basic rule to make the relation transitive.

  \begin{lstlisting}[language=input, caption={Defining a transitive relation with a basic rule}, label={lst:transitive_rule}]
{ p(X,Y) } :- q(X), q(Y).
p(X,Y) :- p(X,Z), p(Z,Y), q(X), q(Y), q(Z).
  \end{lstlisting}

  An alternative program to achieve this is given in Listing~\ref{lst:transitive_constraint}.
  Again, a choice rule is used to define the relation.
  However, to ensure that this relation is transitive a constraint is used this time.

  \begin{lstlisting}[language=input, caption={Defining a transitive relation with a constraint}, label={lst:transitive_constraint}]
{ p(X,Y) } :- q(X), q(Y).
:- p(X,Z), p(Z,Y), not p(X,Y), q(X), q(Y), q(Z).
  \end{lstlisting}

  Again, it is not possible to verify the strong equivalence of these two programs using the current version of \anthem{}.
  While anthem produces an output in this example, it also prints a message stating that the inputs are \emph{non-positive programs}, from which we can conclude that the semantics of the output is the logic of here-and-there.
  This is because the translation implemented in \anthem{} uses the logic of here-and-there as an intermediate step towards transforming the programs into formulas in classical logic.
  For positive logic programs, the semantics of the logic of here-and-there coincides with classical logic, but for programs containing \emph{choices} or \emph{negation}, this is not the case.

  The goal of this thesis is to overcome these current limitations of \anthem{} to enable the verification of the strong equivalence of programs such as Programs~\ref{lst:colour_atoms} and~\ref{lst:colour_pool} as well as Programs~\ref{lst:transitive_rule} and~\ref{lst:transitive_constraint}.
  The main obstacle to achieving this is to map the semantics of the logic of here-and-there to the semantics of classical logic.
  To do so the transformation~$\sigma^*$ is presented, which is based on a transformation defined in~\cite{PearceEtAl2001EncodingsEquilibriumLogic}.
  The applicability of this transformation in the context of strong equivalence is formalised by a theorem in this thesis.

  Furthermore, this thesis extends the translation~$\tau^*$ implemented by \anthem{} to handle programs containing pools.

  Another contribution of this thesis is a formalisation of how the extended translation~$\tau^*$ together with the transformation~$\sigma^*$, can be used to encode the question of strong equivalence of two logic programs in classical first-order logic.
  Both the extended translation~$\tau^*$ and the transformation~$\sigma^*$ are implemented in a new version of \anthem{}\footnote{https://github.com/janheuer/anthem-1/releases/tag/v0.3}.

  As a secondary goal, this thesis compares the ability of several automated theorem provers to verify strong equivalence problems generated by \anthem{}.

  The structure of the thesis is as follows.
  Chapter~\ref{chap:background} formalises the concept of strong equivalence and details the necessary background regarding the logic of here-and-there, as well as how to express the semantics of the logic of here-and-there in classical logic.
  Chapter~\ref{chap:translation} extends the translation of \anthem{} by pools and shows how the transformation from the logic of here-and-there to classical logic can be used to verify the strong equivalence of programs with negation and simple choices.
  A brief description of how this is implemented in \anthem{} is given in Chapter~\ref{chap:implementation}.
  Chapter~\ref{chap:experiments} first introduces several automated theorem provers, before comparing them for verifying the strong equivalence of various logic programs.
  Finally, Chapter~\ref{chap:conclusion} summarises the contributions of this thesis in the context of related work and details some ideas for future work.

  \chapter{Background}
\label{chap:background}
  This chapter first gives a formal definition of the strong equivalence of two logic programs, as well as some context on other forms of equivalence in Answer Set Programming.

  Afterwards, the logic of here-and-there is introduced, which provides a simple characterisation of strong equivalence.
  This logic is used as an intermediate step to expressing the strong equivalence of two logic programs in classical logic.

\section{Strong Equivalence}
\label{sec:strong_eq}
  There are several notions of equivalence in Answer Set Programming.
  The simplest one is the basic equivalence of two programs.

  \begin{definition}[Equivalence]
    Two programs, $\Pi_1$~and~$\Pi_2$ are equivalent if and only if they have the same answer sets.
  \end{definition}

  As an example, the program
  \begin{equation*}
    \begin{array}{l}
      p(X) \leftarrow q(X), \\
      q(1)
    \end{array}
  \end{equation*}
  is equivalent to
  \begin{equation*}
    \begin{array}{l}
      q(1), \\
      p(1),
    \end{array}
  \end{equation*}
  as both have the answer set $\{p(1), q(1)\}$.
  While the two programs have the same answer set, they are not the same, which can be observed by adding the simple fact~$q(2)$.
  The answer set of the first program is now $\{p(1), p(2), q(1), q(2)\}$, while the answer set of the second program does not contain~$p(2)$.

  This idea of programs being the same while adding facts to them is captured by uniform equivalence~\cite{EiterFink2003UniformEquivalenceLogic}.

  \begin{definition}[Uniform Equivalence]
    Two programs, $\Pi_1$~and~$\Pi_2$ are uniformly equivalent if and only if for every set of facts~$F$, $\Pi_1 \cup F$ and $\Pi_2 \cup F$ have the same answer set.
  \end{definition}

  For example, the program
  \begin{equation}
    \begin{array}{l}
      p \leftarrow q, \\
      p \leftarrow \mynot q
    \end{array}
    \label{ex:uniform_a}
  \end{equation}
  is uniformly equivalent to
  \begin{equation}
    p.
    \label{ex:uniform_b}
  \end{equation}
  Getting different answer sets from the two programs cannot be accomplished by adding facts to them.
  However, this can be achieved by adding the rule $q \leftarrow p$ to both programs.
  Program~\ref{ex:uniform_a} does not have any answer sets with this rule, while Program~\ref{ex:uniform_b} has the answer set $\{p, q\}$.

  This leads to an even stricter form of equivalence called strong equivalence~\cite{LifschitzEtAl2001StronglyEquivalentLogic}.

  \begin{definition}[Strong Equivalence]
    Two programs, $\Pi_1$~and~$\Pi_2$ are strongly equivalent if and only if for every program~$\Pi$, the answer sets of $\Pi_1 \cup \Pi$ and $\Pi_2 \cup \Pi$ are the same.
  \end{definition}

  For example, the program
  \begin{equation*}
    \begin{array}{l}
      p \leftarrow q, \\
      q
    \end{array}
  \end{equation*}
  and the program obtained by replacing the rule with a fact
  \begin{equation*}
    \begin{array}{l}
      p, \\
      q
    \end{array}
  \end{equation*}
  are strongly equivalent.
  Removing the constraint from the following program
  \begin{equation}
    \begin{array}{l}
      p \leftarrow q, \\
      \leftarrow \mynot p \land q
    \end{array}
    \label{ex:strong_eq2_a}
  \end{equation}
  also results in a strongly equivalent program
  \begin{equation}
    p \leftarrow q.
    \label{ex:strong_eq2_b}
  \end{equation}

  Both uniform and strong equivalence are interesting concepts in the context of program optimisation.
  If the optimised version of a program is uniformly equivalent to the original, it will lead to the same results on arbitrary inputs.
  Using the concept of strong equivalence makes it possible to only look at part of a program in the optimisation process.
  This is a desirable property, as it is common to split up a program into different parts.
  By verifying the strong equivalence of the optimised part to the original, it is ensured that adding it back to the rest of the program will not change the results on arbitrary inputs.

\section{Logic of Here-and-There}
  The logic of here-and-there is an intermediate logic between classical and intuitionistic logic.
  It was first connected to Answer Set Programming by~\cite{Pearce1997NewLogicalCharacterisation}, who introduced a new logical characterisation of answer sets as a form of minimal models in the logic of here-and-there.
  Furthermore, the logic of here-and-there also provides an alternative characterisation of strong equivalence, which is introduced in Section~\ref{sec:strong_eq_ht}.

  Using this characterisation is still not enough to use automated theorem provers for verifying strong equivalence, as they do not support the semantics of the logic of here-and-there.
  However, it is possible to map the semantics of the logic of here-and-there to the semantics of classical logic, as explained in Section~\ref{sec:mapping}.
\subsection{Definition}
\label{sec:definition}
  The logic of here-and-there uses a standard propositional language built from a set of \emph{atoms}~$\atoms$ as well as the usual logical symbols $\bot, \neg, \land, \lor,$ and $\rightarrow$.
  The \emph{logical complexity} of a formula~$\phi$, written as~$lc(\phi)$, is the number of occurrences of the logic symbols $\neg, \land, \lor,$ and $\rightarrow$ in~$\phi$.
  An \emph{HT-interpretation} is an ordered pair~$\htint{H}{T}$ of sets of atoms such that $H \subseteq T$.

  Intuitively, an interpretation describes two worlds, the ``here'' and ``there''.
  Atoms true in the ``here'' world are provably true, atoms true in the ``there'' world are possibly true, and atoms which are neither true ``here'' nor ``there'' are verified as false.
  Furthermore, everything that is true in the ``here'' world is also true in the ``there'' world.

  The definition of satisfiability in the logic of here-and-there is in parts analogous to classical logic; however, it significantly differs in the definition of negation and implication.
  \begin{definition}[Satisfiability in the Logic of Here-And-There]
    The satisfiability relation~$\models$ is recursively defined as follows:
    \begin{itemize}
      \item $\htint{H}{T} \not\models \bot$
      \item for any atom $p \in \atoms$, $\htint{H}{T} \models p$ if $p \in H$
      \item $\htint{H}{T} \models \phi_1 \land \phi_2$ if $\htint{H}{T} \models \phi_1$ and $\htint{H}{T} \models \phi_2$
      \item $\htint{H}{T} \models \phi_1 \lor \phi_2$ if $\htint{H}{T} \models \phi_1$ or $\htint{H}{T} \models \phi_2$
      \item $\htint{H}{T} \models \neg\phi$ if $\htint{T}{T} \not\models \phi$
      \item $\htint{H}{T} \models \phi_1 \rightarrow \phi_2$ if $ \!
        \begin{aligned}[t]
            &\htint{H}{T} \models \phi_1 \text{ implies } \htint{H}{T} \models \phi_2 \text{ and} \\
            &\htint{T}{T} \models \phi_1 \text{ implies } \htint{T}{T} \models \phi_2.\\
        \end{aligned}$
    \end{itemize}
    \label{def:satisfiability}
  \end{definition}
  Note that $\htint{T}{T} \models \phi$ is equivalent to $T \models \phi$ (in classical logic).

\subsection{Strong Equivalence in the Logic of Here-And-There}
\label{sec:strong_eq_ht}
  In order to characterise strongly equivalent logic programs in the logic of here-and-there, logic programs first have to be represented in the logic of here-and-there.
  For simple programs such as the examples from Section~\ref{sec:strong_eq} this is relatively straightforward.
  For example, Program~\ref{ex:strong_eq2_a} is represented by propositional formulas as follows:
  \begin{equation*}
    \begin{array}{l}
      q \rightarrow p, \\
      \neg p \land q \rightarrow \bot.
    \end{array}
  \end{equation*}

  The following theorem from~\cite[Theorem~1]{LifschitzEtAl2001StronglyEquivalentLogic} reduces the question of strong equivalence to a satisfiability problem in the logic of here-and-there.
  \begin{theorem}[Strong Equivalence]
    Two programs, $\Pi_1$~and~$\Pi_2$ are strongly equivalent if and only if their representations as propositional formulas are equivalent in the logic of here-and-there.
    \label{thm:strong_equivalence}
  \end{theorem}

  The following formula can, for example, express the strong equivalence between Programs~\ref{ex:strong_eq2_a}~and~\ref{ex:strong_eq2_b} in the logic of here-and-there:
  \begin{equation*}
    ((q \rightarrow p) \land (\neg p \land q \rightarrow \bot)) \leftrightarrow (p \rightarrow q).
  \end{equation*}

\subsection{Expressing the Semantics of the Logic of Here-And-There in Classical Logic}
\label{sec:mapping}
  The main idea of expressing satisfiability in the logic of here-and-there in classical logic is to introduce a new set of atoms which represent the values of atoms in the ``there'' world.
  Formally, given a set of atoms~$\atoms$, a new set of atoms~$\atoms^\prime$ is introduced as $\atoms^\prime = \{p^\prime \mid p \in \atoms\}$, which is disjoint to $\atoms$.
  The formula~$\phi^\prime$ is the result of replacing every atom~$p$ in~$\phi$ with~$p^\prime$.
  Given any set $T \subseteq \atoms$, $T^\prime$~is defined as $T^\prime = \{p^\prime \mid p \in T\}$.
  The ``primed'' formulas~$\phi^\prime$ will then be used to express the value of a formula~$\phi$ in the ``there'' world (i.e. the relation $\htint{T}{T} \models \phi$).

  Satisfiability in the logic of here-and-there can then be expressed in classical logic by the following transformation~\cite[Definition~1]{PearceEtAl2001EncodingsEquilibriumLogic}:
  \begin{definition}[$\sigma$]
    Let~$\phi$ be a formula. Then, $\sigma(\phi)$~is recursively defined as follows:
    \begin{itemize}
      \item if $\phi \in \atoms \cup \{\top, \bot\}$, then $\sigma(\phi) = \phi$
      \item if $\phi = \left(\phi_1 \circ \phi_2\right)$, for $\circ \in \{\land, \lor\}$, then $\sigma(\phi) = \sigma(\phi_1) \circ \sigma(\phi_2)$
      \item if $\phi = \neg\psi$, then $\sigma(\phi) = \neg\sigma(\psi) \land \neg\psi^\prime$
      \item if $\phi = \left(\phi_1 \rightarrow \phi_2\right)$, then $\sigma(\phi) = \left(\sigma(\phi_1) \rightarrow \sigma(\phi_2)\right) \land \left(\phi_1^\prime \rightarrow \phi_2^\prime\right)$.
    \end{itemize}
    \label{def:mapping}
  \end{definition}
  Intuitively the unprimed formulas represent the formulas evaluated in the ``here'' world, while the primed formulas represent the ones evaluated in the ``there'' world.
  As in the definition of satisfiability (Definition~\ref{def:satisfiability}) only negated formulas and implications directly depend on the ``there'' world.

  Additionally, the condition $H \subseteq T$, i.e. every atom true ``here'' is also true ``there'', has to be encoded.
  This is done by adding the following set of formulas: $\paxioms = \{p \rightarrow p^\prime \mid p \in \atoms\}$.

  As a consequence of adding $\paxioms$, it also holds that: $\psi \rightarrow \psi^\prime$ for every formula~$\psi$~\cite[Proposition~1]{PearceEtAl2001EncodingsEquilibriumLogic}.
  Using this fact (or rather its contraposition $\neg\psi^\prime \rightarrow \neg\psi$) the transformation from Definition~\ref{def:mapping} can be simplified to the following:

  \begin{definition}[$\sigma^*$]
    Let $\phi$ be a formula. Then, $\sigma^*(\phi)$ is recursively defined as follows:
    \begin{itemize}
      \item if $\phi \in \atoms \cup \{\top, \bot\}$, then $\sigma^*(\phi) = \phi$
      \item if $\phi = \left(\phi_1 \circ \phi_2\right)$, for $\circ \in \{\land, \lor\}$, then $\sigma^*(\phi) = \sigma^*(\phi_1) \circ \sigma^*(\phi_2)$
      \item if $\phi = \neg\psi$, then $\sigma^*(\phi) = \neg\psi^\prime$
      \item if $\phi = \left(\phi_1 \rightarrow \phi_2\right)$, then $\sigma^*(\phi) = \left(\sigma^*(\phi_1) \rightarrow \sigma^*(\phi_2)\right) \land \left(\phi_1^\prime \rightarrow \phi_2^\prime\right)$
    \end{itemize}
    \label{def:simplified_mapping}
  \end{definition}

  This definition is similar to~\cite[Definition~2]{PearceEtAl2001EncodingsEquilibriumLogic}.
  However, \cite[Definition~2]{PearceEtAl2001EncodingsEquilibriumLogic}~only defines the transformation on expressions (i.e. formulas only consisting of atoms and the connectives $\land, \lor,$ and $\neg$).

  In order to state how the equivalence of two formulas in the logic of here-and-there can be expressed in classical logic, let us first consider two lemmas.
  The first lemma relates interpretations in the logic of here-and-there to interpretations in classical logic.

  \begin{lemma}
    There exists a 1-to-1 correspondence between interpretations~$\htint{H}{T}$ in the logic of here-and-there and classical interpretations~$I$ over the alphabet $\atoms \cup \atoms^\prime$ such that $I \models \paxioms$.
    \label{lem:1_to_1}
  \end{lemma}

  \begin{proof}
    First, I will define how a classical interpretation can be obtained from an HT-interpretation and vice versa, before showing that this is a 1-to-1 correspondence.
    \smallskip

    Given an HT-interpretation~$\htint{H}{T}$, the classical interpretation~$I$ can be defined as $I = H \cup T^\prime$.
    As $H \subseteq T$ per definition, it holds that $I \models \paxioms$.

    Given a classical interpretation~$I$, $H$~can be defined as $H = I \cap \atoms$ (i.e. all ``unprimed'' atoms in~$I$) and~$T$ as $T = \{p \mid p^\prime \in I \cap \atoms^\prime\}$ (i.e. the ``unprimed'' versions of all ``primed'' atoms in~$I$).
    As $I \models \paxioms$, the condition $H \subseteq T$ follows.
    Therefore~$\htint{H}{T}$ is an interpretation in the logic of here-and-there.
    \smallskip

    To prove that this is a 1-to-1 correspondence, it first needs to be shown that given an HT-interpretation~$\htint{H}{T}$, the following two equations hold:
    \begin{subequations}
      \begin{align}
        H &= I \cap \atoms, \label{eq:ht_to_i1}\\
        T &= \{p \mid p^\prime \in I \cap \atoms^\prime\} \label{eq:ht_to_i2}
      \end{align}
    \end{subequations}
    with $I$ defined as above.
    Applying the definition of~$I$ to~\eqref{eq:ht_to_i1} gives us
    \begin{align*}
      H ={} &I \cap \atoms \\
      ={} &(H \cup T^\prime) \cap \atoms \tag{definition of $I$} \\
      ={} &(H \cap \atoms) \cup (T^\prime \cap \atoms) \\
      ={} &H \cup \varnothing \tag{$T \subseteq \atoms^\prime$ and $\atoms \cap \atoms^\prime = \varnothing$} \\
      ={} &H.
    \end{align*}
    Similarly for~\eqref{eq:ht_to_i2}
    \begin{align*}
      T ={} &\{p \mid p^\prime \in I \cap \atoms^\prime\} \\
      ={} &\{p \mid p^\prime \in (H \cup T^\prime) \cap \atoms^\prime\} \tag{definition of $I$} \\
      ={} &\{p \mid p^\prime \in (H \cap \atoms^\prime) \cup (T^\prime \cap \atoms^\prime)\} \\
      ={} &\{p \mid p^\prime \in \varnothing \cup T^\prime\} \tag{$H \subseteq \atoms$ and $\atoms \cap \atoms^\prime = \varnothing$} \\
      ={} &\{p \mid p^\prime \in T^\prime\} \\
      ={} &T.
    \end{align*}
    Second, it needs to be shown that given a classical interpretation~$I$, the following equation holds:
    \begin{equation}
      I = H \cup T^\prime \label{eq:i_to_ht}
    \end{equation}
    with~$H$ and~$T$ defined as above.
    Applying the definitions to~\eqref{eq:i_to_ht} gives us:
    \begin{align*}
      I ={} &H \cup T^\prime \\
      ={} &(I \cap \atoms) \cup (\{p \mid p^\prime \in I \cap \atoms^\prime\})^\prime \tag{definition of $H$ and $T$} \\
      ={} &(I \cap \atoms) \cup (\{p^\prime \mid p^\prime \in I \cap \atoms^\prime\}) \\
      ={} &(I \cap \atoms) \cup (I \cap \atoms^\prime) \\
      ={} &I \cap (\atoms \cup \atoms^\prime) \tag{$I \subseteq \atoms \cup \atoms^\prime$} \\
      ={} &I.
    \end{align*}

    Thus, the above defined correspondence is a 1-to-1 correspondence.
  \end{proof}

  The second lemma relates satisfiability in the logic of here-and-there to satisfiability in classical logic.

  \begin{lemma}
    Let~$\phi$ be a formula.
    An HT-interpretation~$\htint{H}{T}$ satisfies~$\phi$, if and only if the classical interpretation $I = H \cup T^\prime$ (over the alphabet $\atoms \cup \atoms^\prime$) satisfies~$\sigma^*(\phi)$.
    \label{lem:satisfiability}
  \end{lemma}

  \begin{proof}
    By induction on~$lc(\phi)$ (for the definition of~$lc$ see Section~\ref{sec:definition}).
    \smallskip

    \textsc{Induction base}.
    Assume $lc(\phi) = 0$.
    Then $\sigma^*(\phi) = \phi$.
    If $\phi = \bot$ or $\phi = \top$ the statement trivially holds.
    If $\phi = p$ for some $p \in \atoms$, then $\htint{H}{T} \models p$ iff $I \models p$ as $H \subseteq I$.
    \smallskip

    \textsc{Induction step}.
    Assume $lc(\phi) > 0$, and let the statement hold for all formulas~$\psi$ such that $lc(\psi) < lc(\phi)$.
    To show that the statement also holds for~$\phi$ several cases have to be considered:

    \emph{Case 1.} $\phi = \phi_1 \circ \phi_2$, for $\circ \in \{\land, \lor\}$

    Then $\sigma^*(\phi) = \sigma^*(\phi_1) \circ \sigma^*(\phi_2)$.
    As $lc(\phi_i) < lc(\phi)$, by induction hypothesis it holds that $\htint{H}{T} \models \phi_i$ iff $I \models \sigma^*(\phi_i)$, for $i=1,2$.
    Therefore, $\htint{H}{T} \models \phi_1 \circ \phi_2$ iff $I \models \sigma^*(\phi_1) \circ \sigma^*(\phi_2)$.

    \emph{Case 2.} $\phi = \neg\psi$

    Then $\sigma^*(\phi) = \neg\psi^\prime$.
    Per Definition~\ref{def:satisfiability},
    \begin{align*}
      \htint{H}{T} \models \neg\psi \leftrightarrow{} &\htint{T}{T} \not\models \psi \\
      \leftrightarrow{} &T \not\models \psi \\
      \leftrightarrow{} &T^\prime \not\models \psi^\prime \\
      \leftrightarrow{} &I \not\models \psi^\prime \tag{$I = H \cup T^\prime$ and $H \cap \atoms^\prime = \varnothing$} \\
      \leftrightarrow{} &I \models \neg\psi^\prime \tag{definition satisfiability}
    \end{align*}
    so it holds that $\htint{H}{T} \models \neg\psi$ iff $I \models \neg\psi^\prime$.

    \emph{Case 3.} $\phi = \phi_1 \rightarrow \phi_2$

    Then $\sigma^*(\phi) = (\sigma^*(\phi_1) \rightarrow \sigma^*(\phi_2)) \land (\phi_1^\prime \rightarrow \phi_2^\prime)$.
    Per Definition~\ref{def:satisfiability},
    \begin{align*}
      \htint{H}{T} \models \phi_1 \rightarrow \phi_2 \leftrightarrow{} &\htint{H}{T} \models \phi_1 \text{ implies } \htint{H}{T} \models \phi_2 \text{ and } \\
      &\htint{T}{T} \models \phi_1 \text{ implies } \htint{T}{T} \models \phi_2.
    \end{align*}
    As $lc(\phi_i) < lc(\phi)$, for $i=1,2$,
    \begin{align*}
      &\htint{H}{T} \models \phi_1 \text{ implies } \htint{H}{T} \models \phi_2 \\
      \leftrightarrow{} &I \models \sigma^*(\phi_1) \text{ implies } I \models \sigma^*(\phi_2)  \tag{induction hypothesis} \\
      \leftrightarrow{} &I \models \sigma^*(\phi_1) \rightarrow \sigma^*(\phi_2) \tag{definition satisfiability}.
    \end{align*}
    Furthermore,
    \begin{align*}
      &\htint{T}{T} \models \phi_i \\
      \leftrightarrow{} &T \models \phi_i \\
      \leftrightarrow{} &T^\prime \models \phi_i^\prime \\
      \leftrightarrow{} &I \models \phi_i^\prime, \tag{$I = H \cup T^\prime$ and $H \cap \atoms^\prime = \varnothing$}
    \end{align*}
    for $i=1,2$.
    Therefore,
    \begin{align*}
      &\htint{T}{T} \models \phi_1 \text{ implies } \htint{T}{T} \models \phi_2 \\
      \leftrightarrow{} &I \models \phi_1^\prime \text{ implies } I \models \phi_2^\prime \tag{see above} \\
      \leftrightarrow{} &I \models \phi_1^\prime \rightarrow \phi_2^\prime \tag{definition satisfiability}.
    \end{align*}
    Hence, $\htint{H}{T} \models \phi_1 \rightarrow \phi_2$ iff $I \models (\sigma^*(\phi_1) \rightarrow \sigma^*(\phi_2)) \land (\phi_1^\prime \rightarrow \phi_2^\prime)$.

    Thus, it is shown that the statement holds for all cases of~$\phi$.
  \end{proof}

  Finally, the theorem on expressing the equivalence of two formulas in the logic of here-and-there in classical logic can be stated:
  \begin{theorem}[HT-Equivalence in Classical Logic]
    Let~$\phi_1$~and~$\phi_2$ be formulas.
    The formula $\phi_1 \leftrightarrow \phi_2$ is valid in the logic of here-and-there, if and only if the formula $\sigma^*(\phi_1) \leftrightarrow \sigma^*(\phi_2)$ is satisfied by every classical interpretation~$I$ over the alphabet $\atoms \cup \atoms^\prime$ such that $I \models \paxioms$.
    \label{thm:ht_equivalence}
  \end{theorem}

  \begin{proof}
    The formula $\phi_1 \leftrightarrow \phi_2$ is satisfiable in the logic of here-and-there iff for every HT-interpretation $\htint{H}{T}$, $\htint{H}{T} \models \phi_1$ iff $\htint{H}{T} \models \phi_2$.
    Similarly, the formula $\sigma^*(\phi_1) \leftrightarrow \sigma^*(\phi_2)$ is satisfied by every classical interpretation~$I$ (over the alphabet $\atoms \cup \atoms^\prime$ with $I \models \paxioms$) iff for every classical interpretation~$I$ (over the alphabet $\atoms \cup \atoms^\prime$ such that $I \models \paxioms$), $I \models \sigma^*(\phi_1)$ iff $I \models \sigma^*(\phi_2)$.

    By Lemma~\ref{lem:1_to_1} there is a 1-to-1 correspondence between interpretations~$\htint{H}{T}$ in the logic of here-and-there and classical interpretations~$I$ with $I \models \paxioms$.
    By Lemma~\ref{lem:satisfiability} $\htint{H}{T}$ satisfies~$\phi_i$ iff the corresponding classical interpretation~$I$ satisfies~$\sigma^*(\phi_i)$, for $i=1,2$.
    Therefore, $\htint{H}{T} \models \phi_1$ iff $\htint{H}{T} \models \phi_2$ is equivalent to $I \models \sigma^*(\phi_1)$ iff $I \models \sigma^*(\phi_2)$.
  \end{proof}

  For example, using Theorem~\ref{thm:strong_equivalence} and Theorem~\ref{thm:ht_equivalence} the strong equivalence of Programs~\ref{ex:strong_eq2_a}~and~\ref{ex:strong_eq2_b} can be verified by proving that the formula
  \begin{equation*}
    \begin{gathered}
      ((q \rightarrow p) \land (q^\prime \rightarrow p^\prime) \land (\neg p^\prime \land q \rightarrow \bot) \land (\neg p^\prime \land q^\prime \rightarrow \bot)) \leftrightarrow \\
      ((q \rightarrow p) \land (q^\prime \rightarrow p^\prime)),
    \end{gathered}
  \end{equation*}
  holds in every classical interpretation~$I$ (over the alphabet $\atoms \cup \atoms^\prime$) such that $I \models \{ p \rightarrow p^\prime, q \rightarrow q^\prime \}$.

  \chapter[Translating Logic Programs into Formulas]{Translating Logic Programs into Classical First-Order Formulas}
\label{chap:translation}
  In this chapter, the translation implemented in \anthem{} is defined.

  First, the new input language is defined, which extends the previous input language of \anthem{} by pooling.

  Second, it is shown how these logic programs are translated into first-order formulas.
  This translation again is an extension of the previous translation of \anthem{} in order to handle pooling.

  Finally, the transformation of logic programs to classical first-order formulas is completed by expressing the semantics of the logic of here-and-there in classical logic as defined in Section~\ref{sec:mapping}.
  Furthermore, a theorem formalising the correctness of this transformation is given.

\section{Input Language}
\label{sec:input}
  The definition of the input language in this section extends the input language of the previous version of \anthem{} by pooling.
  Before defining this language formally, let us first look at an informal explanation of pooling.

  Pooling is a feature similar to intervals.
  Both are shorthand notations allowing one to express a set of values in a single term.
  For example, to express the set $\{1,...,3\}$ it is possible to use both a pool $1;2;3$ or an interval $1 \twodots 3$.

  However, the difference is that intervals only make it possible to express a set of consecutive integers.
  With pools, it is possible to express non-consecutive integers as well as non-numerical values.
  As an example, to express the even numbers between 1 and 10 the following pooling expression can be used: $2;4;6;8;10$ (to do so using intervals it is necessary to combine an interval with an arithmetic expression: $2 \times (1 \twodots 5)$).
  Another example of a pool is the following:
  \begin{equation}
    a;1 \twodots 5.
    \label{eq:mixed_pool}
  \end{equation}
  Here, a non-numerical value $a$ is part of the pool as well as a set of numerical values expressed by an interval.

  Using pools in atoms in either the head or the body is again similar to using an interval.
  The rule $p \leftarrow q(1 \twodots 3)$ has the same meaning as the rule
  \begin{equation}
    p \leftarrow q(1;2;3),
    \label{eq:pool_body}
  \end{equation}
  $q$~will be true if at least one of~$q(1), q(2),$ or~$q(3)$ is true.
  Using a pool (or an interval) in the head has the meaning that the atom will be true for all values if the body is true.
  For example the rule
  \begin{equation}
    p(X;Y) \leftarrow q(X,Y)
    \label{eq:pool_head}
  \end{equation}
  says that both~$p(X)$~and~$p(Y)$ hold if~$q(X,Y)$ holds.
  Without pools, it would be necessary to express this using two rules (as the current language does not include general choice rules).
  An interval cannot be used here, as the values~$X$~and~$Y$ are neither necessarily consecutive nor numerical.

  The definition of the input language in this thesis corresponds to the one described in~\cite[Section~2]{LifschitzEtAl2019VerifyingStrongEquivalence} with the exception that pools and tuples are allowed as program terms, and the definition of atoms is extended by pools.
  More formally, the following two alternatives are added to the recursive definition of \emph{program terms}:

  \begin{itemize}
    \item if $t_1, \dots, t_k$ are program terms then $(t_1, \dots, t_k)$ is a program term
    \item if $t_1, \dots, t_k$ are program terms then $(t_1; \dots; t_k)$ is a program term
  \end{itemize}
  The definition of an \emph{atom} is extended by the following:

  \begin{itemize}
    \item $p(\mathbf{t_1}; \dots; \mathbf{t_k})$ is an atom, where~$p$ is a symbolic constant and each~$\mathbf{t_i}$ is a tuple of program terms
  \end{itemize}

  Notably, compared to the language defined in~\cite{GebserEtAl2015AbstractGringo} (which the definition in~\cite[Section~2]{LifschitzEtAl2019VerifyingStrongEquivalence} is based on), here both tuples and pools are part of the program terms.
  The semantics of this language is the usual stable model semantics (for details, see~\cite[Section 3]{LifschitzEtAl2019VerifyingStrongEquivalence} or~\cite{GebserEtAl2015AbstractGringo}).

\section[Transforming Logic Programs into Formulas]{Transforming Logic Programs into First-Order Formulas}
\label{sec:translation}
  The definition of the translation~$\tau^*$ in this section corresponds to the definition from~\cite[Section~6]{LifschitzEtAl2019VerifyingStrongEquivalence} and is extended to cover the new input language.

  The target language of this translation is a standard first-order language with quantifiers.
  Notably, this language has variables of two sorts: \emph{program variables} and \emph{integer variables}.
  For details see~\cite[Section~5]{LifschitzEtAl2019VerifyingStrongEquivalence}.

  $\tau^*$~is defined using two translations~$\tau^B$~and~$\tau^H$ for the body and head of a rule respectively.
  Before defining these translations, the formula~$\val{t}{Z}$, expressing that~$Z$ is one of the values of the program term~$t$, needs to be defined.
  This is necessary, as in the input language a term can express a set of values (e.g. by using an interval or a pool), whereas the target language does not include sets.
  Furthermore, the arithmetic operations~$/$~and~$\backslash$ have to be encoded.

  The following definition of $\val{t}{Z}$ is extended by tuples and pools compared to the definition in~\cite[Section~6]{LifschitzEtAl2019VerifyingStrongEquivalence}.

  \begin{definition}[$\val{t}{Z}$]
    For every program term~$t$ the formula~$\val{t}{Z}$, where~$Z$ is a program variable that does not occur in~$t$, is recursively defined as follows:
    \begin{itemize}
      \item if $t$ is a numeral, a symbolic constant, a program variable, $\myinf$, or $\mysup$, then $\val{t}{Z}$ is $Z = t$
      \item if $t$ is $(t_1 \circ t_2)$ where $\circ \in \{+, -, \times\}$, then $\val{t}{Z}$ is
        \[
          \exists I J (Z = I \circ J \land \val{t_1}{I} \land \val{t_2}{J}),
        \]
        where $I, J$ are fresh integer variables
      \item if $t$ is $(t_1 / t_2)$, then $\val{t}{Z}$ is
        \begin{align*}
          \exists I J Q R (I = J \times Q + R &\land \val{t_1}{I} \land \val{t_2}{J} \\
          &\land J \neq 0 \land R \geq 0 \land R < Q \land Z = Q),
        \end{align*}
        where $I, J, Q, R$ are fresh integer variables
      \item if $t$ is $(t_1 \backslash t_2)$, then $\val{t}{Z}$ is
        \begin{align*}
          \exists I J Q R (I = J \times Q + R &\land \val{t_1}{I} \land \val{t_2}{J} \\
          &\land J \neq 0 \land R \geq 0 \land R < Q \land Z = R),
        \end{align*}
        where $I, J, Q, R$ are fresh integer variables
      \item if $t$ is $(t_1 \twodots t_2)$, then $\val{t}{Z}$ is
        \[
          \exists I J K (\val{t_1}{I} \land \val{t_2}{J} \land I \leq K \land K \leq J \land Z = K),
        \]
        where $I, J, K$ are fresh integer variables
      \item if $t$ is $(t_1, \dots, t_k)$, then $\val{t}{Z}$ is
        \[
          \exists I_1 \dots I_k (Z = (I_1, \dots, I_k) \land \val{t_1}{I_1} \land \dots \land \val{t_k}{I_k}),
        \]
        where $I_1, \dots, I_k$ are fresh program variables
      \item if $t$ is $(t_1; \dots; t_k)$, then $\val{t}{Z}$ is
        \[
          \exists I_1 \dots I_k (\val{t_1}{I_1} \land \dots \land \val{t_k}{I_k} \land (Z = I_1 \lor \dots \lor Z = I_k)),
        \]
        where $I_1, \dots, I_k$ are fresh program variables
    \end{itemize}
  \end{definition}

  For example, if~$t$ is the pool from Listing~\ref{lst:colour_pool}, i.e.~$r;g;b$, then~$\val{t}{Z}$ is
  \[
    \exists I_1 I_2 I_3 (\val{r}{I_1} \land \val{g}{I_2} \land \val{b}{I_3} \land (Z = I_1 \lor Z = I_2 \lor Z = I_3)).
  \]
  Applying the definition of~$\val{t}{Z}$ recursively gives the following formula:
  \[
    \exists I_1 I_2 I_3 (I_1 = r \land I_2 = g \land I_3 = b \land (Z = I_1 \lor Z = I_2 \lor Z = I_3)).
  \]
  Simplifying this formula to
  \[
    Z = r \lor Z = g \lor Z = b
  \]
  makes it evident that this states exactly that~$Z$ is a value of the pool~$r;g;b$.
  The following example illustrates the necessity to make the definition of~$\val{t}{Z}$ recursive.
  If~$t$ is the Expression~\ref{eq:mixed_pool}, i.e.~$a;1 \twodots 5$, then~$\val{t}{Z}$ is
  \[
    \exists I_1 I_2 (\val{a}{I_1} \land \val{1 \twodots 5}{I_2} \land (Z = I_1 \lor Z = I_2)).
  \]
  Here one of the sub-terms of the pool is itself an expression which represents multiple values: the interval~$1 \twodots 5$.

  Next is the definition of the translation~$\tau^B$ applied to the expressions in the body of a rule.
  Compared to the definition in~\cite[Section~6]{LifschitzEtAl2019VerifyingStrongEquivalence} the following definition is extended to handle atoms containing pools.

  \begin{definition}[$\tau^B$]\hfill
    \begin{itemize}
      \item $\tau^B(p(t_1, \dots, t_k))$ is
        \[
          \exists Z_1 \dots Z_k (\val{t_1}{Z_1} \land \dots \land \val{t_k}{Z_k} \land p(Z_1, \dots, Z_k))
        \]
      \item $\tau^B(\mynot p(t_1, \dots, t_k))$ is
        \[
          \exists Z_1 \dots Z_k (\val{t_1}{Z_1} \land \dots \land \val{t_k}{Z_k} \land \neg p(Z_1, \dots, Z_k))
        \]
      \item $\tau^B(\mynot \mynot p(t_1, \dots, t_k))$ is
        \[
          \exists Z_1 \dots Z_k (\val{t_1}{Z_1} \land \dots \land \val{t_k}{Z_k} \land \neg \neg p(Z_1, \dots, Z_k))
        \]
      \item $\tau^B(p(t_1; \dots; t_k))$ is
        \[
          \tau^B(p(t_1)) \lor \dots \lor \tau^B(p(t_k))
        \]
      \item $\tau^B(\mynot p(t_1; \dots; t_k))$ is
        \[
          \tau^B(\mynot p(t_1)) \lor \dots \lor \tau^B(\mynot p(t_k))
        \]
      \item $\tau^B(\mynot \mynot p(t_1; \dots; t_k))$ is
        \[
          \tau^B(\mynot \mynot p(t_1)) \lor \dots \lor \tau^B(\mynot \mynot p(t_k))
        \]
      \item $\tau^B(t_1 \sim t_2)$ where $\sim \in \{=, \neq, <, >, \leq, \geq\}$ is
        \[
          \exists Z_1 Z_2 (\val{t_1}{Z_1} \land \val{t_2}{Z_2} \land Z_1 \sim Z_2)
        \]
    \end{itemize}
    where each~$Z_i$ is a fresh program variable.
  \end{definition}

  For example, applying~$\tau^B$ to the atom $q(1;2;3)$ from Rule~\ref{eq:pool_body} results in the following formula:
  \[
    \tau^B(q(1)) \lor \tau^B(q(2)) \lor \tau^B(q(3)).
  \]
  Applying the definition of $\tau^B$ recursively gives
  \[
    \exists Z_1 (\val{1}{Z_1} \land q(Z_1)) \lor \exists Z_1 (\val{2}{Z_2} \land q(Z_2)) \lor \exists Z_1 (\val{3}{Z_3} \land q(Z_3))).
  \]
  Simplifying this formula results in
  \[
    q(1) \lor q(2) \lor q(3),
  \]
  which expresses exactly that the atom $q(1;2;3)$ in the body of a rule is satisfied if at least one of~$q(1), q(2),$ or~$q(3)$ is satisfied.

  The translation~$\tau^H$, which is applied to expressions in the head of a rule, is defined as follows.
  Again, the definition is extended to handle atoms containing pools compared to the definition in~\cite[Section~6]{LifschitzEtAl2019VerifyingStrongEquivalence}.

  \begin{definition}[$\tau^H$]\hfill
    \begin{itemize}
      \item $\tau^H(p(t_1, \dots, t_k))$ is
        \[
          \forall Z_1 \dots Z_k (\val{t_1}{Z_1} \land \dots \land \val{t_k}{Z_k} \rightarrow p(Z_1, \dots, Z_k))
        \]
      \item $\tau^H(p(t_1; \dots; t_k))$ is
        \[
          \tau^H(p(t_1)) \land \dots \land \tau^H(p(t_k))
        \]
      \item $\tau^H(\{p(t_1, \dots, t_k)\})$ is
        \[
          \forall Z_1 \dots Z_k (\val{t_1}{Z_1} \land \dots \land \val{t_k}{Z_k} \rightarrow p(Z_1, \dots, Z_k) \lor \neg p(Z_1, \dots, Z_k))
        \]
      \item $\tau^H(\{p(t_1; \dots; t_k)\})$ is
        \[
          \tau^H(\{p(t_1)\}) \land \dots \land \tau^H(\{p(t_k)\})
        \]
      \item $\tau^H()$ (the empty head) is $\perp$
    \end{itemize}
    where each~$Z_i$ is a fresh program variable.
  \end{definition}

  For example, applying~$\tau^H$ to the head of Rule~\ref{eq:pool_head} results in
  \[
    \tau^H(p(X)) \land \tau^H(p(Y)).
  \]
  Applying the definition of~$\tau^H$ recursively and simplifying the formula gives
  \[
    p(X) \land p(Y)
  \]
  which states exactly that if the body of Rule~\ref{eq:pool_head} is satisfied, both~$p(X)$~and~$p(Y)$ have to be true.

  Finally, $\tau^*$~can be defined using these components.
  \begin{definition}[$\tau^*$]
    \[
      \tau^*(H \leftarrow B_1 \land \dots \land B_n)
    \]
    is defined as the universal closure of the formula
    \[
      \tau^B(B_1) \land \dots \land \tau^B(B_n) \rightarrow \tau^H(H).
    \]
  \end{definition}
  For any program~$\Pi$, $\tau^* \Pi$ is the set of formulas $\tau^* R$ for every rule~$R$~in~$\Pi$.

  For example, the result of applying~$\tau^*$ to the following program:
  \begin{equation}
    \begin{array}{l}
      \{ p \} \\
      :- \mynot p, q
    \end{array}
    \label{eq:choice_negation}
  \end{equation}
  are the following two formulas:
  \begin{equation*}
    \begin{array}{l}
      \top \rightarrow (p \lor \neg p) \\
      (\neg p \land q) \rightarrow \bot
    \end{array}
  \end{equation*}
  Applying~$\tau^*$ to Program~\ref{eq:pool_head} gives
  \begin{equation*}
    \begin{aligned}
      \forall U_1 U_2 (&\exists Z_1 Z_2 (Z_1 = U_1 \land Z_2 = U_2 \land q(Z_1,Z_2)) \rightarrow \\
      &(\forall Z_3 (Z_3 = U_1 \rightarrow p(Z_3)) \land \forall Z_4 (Z_4 = U_2 \rightarrow p(Z_4))))
    \end{aligned}
  \end{equation*}

\section{Expressing Strong Equivalence in Classical Logic}
\label{sec:semantics_map}
  After applying~$\tau^*$, the resulting formulas will still have the semantics of the logic of here-and-there.
  Therefore, it is still not possible to verify the strong equivalence of two programs using theorem provers for classical first order logic.
  An exception is the case of two positive logic programs, as in that case, the semantics of the logic of here-and-there coincides with classical logic.
  This is the basis for the previous version of \anthem{}~\cite{LifschitzEtAl2019VerifyingStrongEquivalence}.

  However, with the transformation~$\sigma^*$ given by Definition~\ref{def:mapping}, it is possible to reduce the problem of verifying strong equivalence to classical first-order logic, by expressing the semantics of the formulas from the logic of here-and-there in classical logic.
  This idea is formalised in the remainder of this section.

  The notation~$F^\text{prop}$ is used to signify the infinitary propositional formula corresponding to (the closed formula)~$F$, for details see~\cite[Section~5]{LifschitzEtAl2019VerifyingStrongEquivalence}.

  The following lemma relates the strong equivalence of two logic programs to the strong equivalence of their translation obtained by applying~$\tau^*$:

  \begin{lemma}
    A program~$\Pi_1$ is strongly equivalent to a program~$\Pi_2$ if and only if $(\tau^* \Pi_1)^\text{prop}$ is strongly equivalent to $(\tau^* \Pi_2)^\text{prop}$.
    \label{lem:extension}
  \end{lemma}

  This corresponds to~\cite[Proposition~4]{LifschitzEtAl2019VerifyingStrongEquivalence}.
  As a reminder, the translation~$\tau^*$ in this thesis differs to the one in~\cite[Section~6]{LifschitzEtAl2019VerifyingStrongEquivalence} by the extensions for pools and tuples as program terms as well as atoms containing pools.
  It is clear that on the input language from~\cite[Section~2]{LifschitzEtAl2019VerifyingStrongEquivalence} the two translations coincide.
  Lemma~\ref{lem:extension} then follows from generalising~\cite[Proposition~4]{LifschitzEtAl2019VerifyingStrongEquivalence} to the extensions of the input language.

  It is not possible yet to apply Theorem~\ref{thm:strong_equivalence} (which expresses strong equivalence in the logic of here-and-there) to Lemma~\ref{lem:extension} as the formulas $(\tau^* \Pi_i)^\text{prop}$, ${i=1,2}$, are infinitary formulas.
  In contrast, Theorem~\ref{thm:strong_equivalence} only covers the case of finite formulas.
  However, Theorem~\ref{thm:strong_equivalence} can be generalised to an infinitary logic of here-and-there.
  Applying this generalisation to Lemma~\ref{lem:extension} leads to the following lemma:

  \begin{lemma}
      A program~$\Pi_1$ is strongly equivalent to a program~$\Pi_2$ if and only if $(\tau^* \Pi_1)^\text{prop}$ is equivalent to $(\tau^* \Pi_2)^\text{prop}$ in the logic of here-and-there\footnote{Assuming only standard interpretations, i.e.\ interpretations that interpret equality and arithmetic in the usual way.}.
    \label{lem:extension_ht}
  \end{lemma}

  This lemma already holds for the version of~$\tau^*$ from~\cite[Section~6]{LifschitzEtAl2019VerifyingStrongEquivalence}.
  In order to now apply the transformation~$\sigma^*$, it first needs to be generalised to formulas with quantifiers.
  This is done by extending Definition~\ref{def:simplified_mapping} by the following
  \begin{equation*}
    \begin{gathered}
      \sigma^*(\forall X F(X)) = \forall X \sigma^*(F(X)), \\
      \sigma^*(\exists X F(X)) = \exists X \sigma^*(F(X)).
    \end{gathered}
  \end{equation*}
  Using this generalised~$\sigma^*$, the following theorem can be established:

  \begin{theorem}[Strong Equivalence in Classical Logic]
      A program~$\Pi_1$ is strongly equivalent to a program~$\Pi_2$ if and only if $\sigma^*(\tau^* \Pi_1)$ is equivalent to $\sigma^*(\tau^* \Pi_2)$ in every classical interpretation~$I$\footnotemark[1] (over the alphabet $\atoms \cup \atoms^\prime$) such that $I \models \paxioms$.
    \label{thm:strong_eq_classical}
  \end{theorem}

  The conditions to the classical interpretation~$I$ (i.e. the alphabet $\atoms \cup \atoms^\prime$ and $I \models \paxioms$) ensure that~$I$ corresponds to a valid HT-interpretation.
  This theorem is justified by the following equivalent conditions:

  \begin{enumerate}
    \item $\Pi_1$ is strongly equivalent to $\Pi_2$
    \item $(\tau^* \Pi_1)^\text{prop}$ is equivalent to $(\tau^* \Pi_2)^\text{prop}$ in the logic of here-and-there
    \item $(\sigma^*(\tau^* \Pi_1))^\text{prop}$ is equivalent to $(\sigma^*(\tau^* \Pi_2))^\text{prop}$ in every classical interpretation~$I$ (over the alphabet $\atoms \cup \atoms^\prime$) such that $I \models \paxioms$
    \item $\sigma^*(\tau^* \Pi_1)$ is equivalent to $\sigma^*(\tau^* \Pi_2)$ in every classical interpretation~$I$ (over the alphabet $\atoms \cup \atoms^\prime$) such that $I \models \paxioms$
  \end{enumerate}

  The equivalence of~1 and~2 is precisely Lemma~\ref{lem:extension_ht}.
  The equivalence of~2 and~3 follows from generalising Theorem~\ref{thm:ht_equivalence} (which describes how the equivalence of two formulas in the logic of here-and-there can be expressed in classical logic) to the extension of~$\sigma^*$ to quantifiers.
  The equivalence of~3 and~4 follows from the definition of~$F^\text{prop}$ for a Formula~$F$ (see~\cite[section~5]{LifschitzEtAl2019VerifyingStrongEquivalence}).

  \chapter{Implementation}
\label{chap:implementation}
  The extension of~$\tau^*$ as well as the additional transformation~$\sigma^*$ are implemented in \anthem{}~0.3\footnote{https://github.com/janheuer/anthem-1/releases/tag/v0.3}.

  To obtain the Abstract Syntax Tree (AST) of the given input program \anthem{} makes use of the \clingo{}\footnote{https://github.com/potassco/clingo} API.
  The AST is then transformed into a first-order AST representing the first-order formulas using~$\tau^*$ (implemented using the parts described in Section~\ref{sec:translation}).
  During the transformation, a flag keeps track of the semantics of the formulas.
  This flag is initialised to the semantics of classical logic.
  If a rule contains either negation or a choice this flag is set to the semantics of the logic of here-and-there.

  Depending on how the semantics-flag is set after the transformation defined by~$\tau^*$ terminates (i.e. if it set to the semantics of the logic of here-and-there), $\sigma^*$~is applied.
  Before transforming the formulas, the respective prime atoms (the actual name of the ``primed'' atoms depends on the output format as explained later) and the prime axioms are created.
  The implementation of the actual transformation defined by~$\sigma^*$ (Definition~\ref{def:simplified_mapping}) is done in two steps.
  First, all formulas are duplicated.
  Second, in the first copy all negated atoms are replaced by their ``primed'' variant, and in the second copy all atoms are replaced by their ``primed'' variant.

  This is a slight simplification of~$\sigma^*$ as defined in Definition~\ref{def:simplified_mapping}.
  For example, applying~$\sigma^*$ (as in Definition~\ref{def:simplified_mapping}) to the following formula
  \[
    \forall U_1 (\exists X_1 (X_1 = U_1 \land \neg q(X_1)) \rightarrow \forall X_2 (X_2 = U_1 \rightarrow p(X_2)))
  \]
  (which is the result of applying~$\tau^*$ to the rule $p(X) \leftarrow \mynot q(X)$) results in
  \begin{align*}
    \forall U_1 (\exists X_1 (X_1 = U_1 \land \neg q^\prime(X_1)) \rightarrow \forall X_2 (&(X_2 = U_1 \rightarrow p(X_2)) \land \\
                                                                                           &(X_2 = U_1 \rightarrow p^\prime(X_2))) \land \\
                 \exists X_3 (X_3 = U_1 \land \neg q^\prime(X_3)) \rightarrow \forall X_4 (&(X_4 = U_1 \rightarrow p^\prime(X_4)) \land \\
                                                                                           &(X_4 = U_1 \rightarrow p^\prime(X_4)))).
  \end{align*}
  The implemented version, however, produces the following formula
  \begin{align*}
    &\forall U_1 (\exists X_1 (X_1 = U_1 \land \neg q^\prime(X_1)) \rightarrow \forall X_2 (X_2 = U_1 \rightarrow p(X_2))) \land \\
    &\forall U_2 (\exists X_3 (X_3 = U_2 \land \neg q^\prime(X_3)) \rightarrow \forall X_4 (X_4 = U_2 \rightarrow p^\prime(X_4))).
  \end{align*}
  It is clear that the implemented~$\sigma^*$ is equivalent to Definition~\ref{def:simplified_mapping} (in the context of the considered input language).

  Finally, \anthem{} outputs the obtained first-order formulas.
  If \anthem{} is used with one input program, \anthem{} prints all the formulas obtained from the rules in the program.
  For example, given Program~\ref{lst:colour_pool} \anthem{} produces the following output
  \begin{lstlisting}[language=human-readable]
(#true -> (forall X1 ((X1 = r) -> colour(X1)) and
           forall X2 ((X2 = g) -> colour(X2)) and
           forall X3 ((X3 = b) -> colour(X3))))
  \end{lstlisting}
  Additionally, \anthem{} prints the following info message
  \begin{lstlisting}[language=anthem-output]
info: output semantics: classical logic
  \end{lstlisting}
  indicating the semantics of the output.
  For Program~\ref{eq:choice_negation} \anthem{} prints
  \begin{lstlisting}[language=human-readable]
(p -> p')
(q -> q')
(#true -> (p or not p'))
(#true -> (p' or not p'))
((not p' and q) -> #false)
((not p' and q') -> #false)
  \end{lstlisting}
  as well as the info message
  \begin{lstlisting}[language=anthem-output]
info: mapped to output semantics: classical logic
  \end{lstlisting}
  indicating that~$\sigma^*$ was applied to transform the formulas to be in the semantics of classical logic.
  The ``primed'' atoms are indicated by appending a prime to the respective atom.
  The first two formulas in the output are the prime axioms.

  If \anthem{} is called with two input programs, some additional steps are done before the output.
  First, the formulas of each program are turned into a conjunction.
  Then a formula is constructed which states the equivalence of the two conjunctions.
  For example, calling \anthem{} with Programs~\ref{lst:colour_atoms}~and~\ref{lst:colour_pool} results in the following output
  \begin{lstlisting}[language=human-readable]
(((#true -> forall X1 ((X1 = r) -> colour(X1))) and
  (#true -> forall X2 ((X2 = g) -> colour(X2))) and
  (#true -> forall X3 ((X3 = b) -> colour(X3))))
<->
 ((#true -> (forall X4 ((X4 = r) -> colour(X4)) and
             forall X5 ((X5 = g) -> colour(X5)) and
             forall X6 ((X6 = b) -> colour(X6))))))
  \end{lstlisting}

  In order to verify the strong equivalence of two programs using automated theorem provers, \anthem{} provides the output format TPTP.
  TPTP~\cite{Sutcliffe2009TPTPProblemLibrary} is a library of example problems for theorem provers written in a standard language for theorem provers provided by TPTP.
  Specifically, \anthem{} uses the typed first-order form (TFF~\cite{SutcliffeEtAl2012TPTPTypedFirstOrder}) with integer arithmetic.
  In the TPTP output format \anthem{} not only outputs the formulas corresponding to the input program(s) but also a type definition for every atom.
  Furthermore, auxiliary definitions and axioms are produced to encode the type system of the output language of \anthem{}.
  \bigskip

  For example, in the TPTP format, the formulas representing Program~\ref{eq:choice_negation} are
  \begin{lstlisting}[language=tptp]
$true => (p | (~p__prime__))
$true => (p__prime__ | (~p__prime__))
((~p__prime__) & q) => $false
((~p__prime__) & q__prime__) => $false
  \end{lstlisting}
  here the ``primed'' atoms are represented by appending \lstinline{__prime__} to the respective atom (as a prime cannot be part of the name in the TPTP language).
  The ``prime'' axioms are represented as formulas of the type axiom in the TPTP output.

  The output of \anthem{} can then be used as the input to an automated theorem prover to verify the strong equivalence of the two input programs.
  Some options for automated theorem provers supporting the TFF language with integer arithmetic are presented in Section~\ref{sec:atps}.
  The remainder of Chapter~\ref{chap:experiments} compares the theorem provers on examples of positive programs, programs with pools, programs with negation, and programs with simple choice rules.

  \chapter{Experimental Evaluation}
\label{chap:experiments}
  This chapter first provides an overview of available automated theorem provers which can be used together with \anthem{}.
  Afterwards, these provers are compared for their ability to the verify strong equivalence of different logic programs.
  The example logic programs include positive programs, programs with pools, programs with negation, and programs with simple choice rules.

\section{Automated Theorem Provers}
\label{sec:atps}
  A large number of automated theorem provers support the TPTP language\footnote{A list of some available options: http://www.tptp.org/cgi-bin/SystemOnTPTP}.
  However, the number of provers supporting the TPTP dialect TFF with integer arithmetic is significantly smaller.
  From the available options, some have also not been updated in a long time and trying to install them results in errors.

  The following provers are tested in the remainder of this chapter: \cvc{} version:~1.8\footnote{https://github.com/CVC4/CVC4/releases/tag/1.8}, \princess{} version:~2020-03-12\footnote{http://www.philipp.ruemmer.org/princess-sources.shtml}, \vampire{} version:~4.5.1\footnote{https://github.com/vprover/vampire/releases/tag/4.5.1}, and \zipperposition{} version:~1.5\footnote{https://github.com/sneeuwballen/zipperposition/releases/tag/1.5}.

  All provers were used with a timeout of 300 s (timeout are indicated by ``---'' in the following sections).

  \cvc{} was used with the options \lstinline[breaklines=true]{--lang tptp --stats --tlimit=300000}.
  \princess{} was used with the options \lstinline[breaklines=true]{-inputFormat=tptp -portfolio=casc -timeout=300000 }.
  \vampire{} is the only prover in this comparison which provides parallelisation.
  Therefore, \vampire{} was used in two different configurations: not making use of the parallelisation with the options \lstinline[breaklines=true]{--mode casc --time_limit 300} (denoted by \vampire{}), and using the parallelisation with the additional option \lstinline{--cores 4} (denoted by \vampirep).
  \zipperposition{} was used with the option \lstinline{-timeout 300}.

  For some examples, \cvc{} terminates before the timeout but still does not manage to prove the example.
  In these cases, \cvc{} terminates with the message \lstinline{GaveUp} and the status shows \lstinline{unknown (INCOMPLETE)} (this is indicated by \mbox{``---*''} in the following sections).
  The theorem proving implemented in \cvc{} seems to be incomplete, and \cvc{} recognises when it is not possible to prove a theorem as a result of the incompleteness (the documentation of \cvc{} does not provide any details on this status message).

\section{Examples of Positive Logic Programs}
  First, the provers are compared on some positive logic programs.
  The examples used in this section are all taken from~\cite[Section~9]{LifschitzEtAl2019VerifyingStrongEquivalence}.
  Therefore, only the results of the comparison will be discussed.
  All of the examples are strongly equivalent.

  \vspace{\extop}\noindent
  \begin{tabular}{p{\excol}  p{\excol}}
    \textbf{Program 1.A} & \textbf{Program 1.B} \\
    \begin{tabularlstlisting}[language=input]
p(X+1) :- q(X).
    \end{tabularlstlisting}
    &
    \begin{tabularlstlisting}[language=input]
p(X) :- q(X-1).
    \end{tabularlstlisting}
  \end{tabular}

  \vspace{\extop}\noindent
  \begin{tabular}{p{\excol}  p{\excol}}
    \textbf{Program 2.A} & \textbf{Program 2.B} \\
    \begin{tabularlstlisting}[language=input]
p(X+X) :- q(X).
    \end{tabularlstlisting}
    &
    \begin{tabularlstlisting}[language=input]
p(X+Y) :- q(X), X=Y.
    \end{tabularlstlisting}
  \end{tabular}

  \vspace{\extop}\noindent
  \begin{tabular}{p{\excol}  p{\excol}}
    \textbf{Program 3.A} & \textbf{Program 3.B} \\
    \begin{tabularlstlisting}[language=input]
p(X+X) :- q(X).
    \end{tabularlstlisting}
    &
    \begin{tabularlstlisting}[language=input]
p(2*X) :- q(X).
    \end{tabularlstlisting}
  \end{tabular}

  \vspace{\extop}\noindent
  \begin{tabular}{p{\excol}  p{\excol}}
    \textbf{Program 4.A} & \textbf{Program 4.B} \\
    \begin{tabularlstlisting}[language=input]
p(X+Y) :- q(X), X=Y.
    \end{tabularlstlisting}
    &
    \begin{tabularlstlisting}[language=input]
p(2*X) :- q(X).
    \end{tabularlstlisting}
  \end{tabular}

  \vspace{\extop}\noindent
  \begin{tabular}{p{\excol}  p{\excol}}
    \textbf{Program 5.A} & \textbf{Program 5.B} \\
    \begin{tabularlstlisting}[language=input]
p(X) :- X>3, X<5.
    \end{tabularlstlisting}
    &
    \begin{tabularlstlisting}[language=input]
p(4).
    \end{tabularlstlisting}
  \end{tabular}

  \vspace{\extop}\noindent
  \begin{tabular}{p{\excol}  p{\excol}}
    \textbf{Program 6.A} & \textbf{Program 6.B} \\
    \begin{tabularlstlisting}[language=input]
p(X) :- X<3, X>5.
    \end{tabularlstlisting}
    &
    \begin{tabularlstlisting}[language=input]
q :- q.
    \end{tabularlstlisting}
  \end{tabular}

  \vspace{\extop}\noindent
  \begin{tabular}{p{\excol}  p{\excol}}
    \textbf{Program 7.A} & \textbf{Program 7.B} \\
    \begin{tabularlstlisting}[language=input]
p(X+0).
    \end{tabularlstlisting}
    &
    \begin{tabularlstlisting}[language=input]
p(X+1).
    \end{tabularlstlisting}
  \end{tabular}

  \begin{table}
    \begin{adjustbox}{max width=\textwidth}
      \begin{tabular}{| l *{6}{| c} |}
        \hline
         & \cvc & \princess & \vampire & \vampirep & \zipperposition \\
        \hline
        Example 1 & --- & --- & 6.045 s & 0.042 s & --- \\
        \hline
        Example 2 & 0.018 s & 0.006 s & 0.016 s & 0.008 s & 0.061 s \\
        \hline
        Example 3 & 0.034 s & --- & 58.89 s & 11.017 s & --- \\
        \hline
        Example 4 & 0.034 s & --- & 0.835 s & 1.488 s & --- \\
        \hline
        Example 5 & --- * & --- & 25.39 s & 5.467 s & --- \\
        \hline
        Example 6 & --- * & 24.472 s & 24.636 s & 4.147 s & 0.032 s \\
        \hline
        Example 7 & --- * & --- & 0.126 s & 0.014 s & 0.529 s \\
        \hline
      \end{tabular}
    \end{adjustbox}
    \caption{Results on examples of positive logic programs}
    \label{tab:positive}
  \end{table}

  Table~\ref{tab:positive} shows the results of each theorem prover.
  \vampire{} is the only prover that manages to verify every example.
  Making use of the parallelisation of \vampire{} leads to significant improvements for most examples.
  \cvc{} only manages to verify three of the seven examples but does so faster than \vampire{}.
  \zipperposition{} also manages to verify three examples.
  \princess{} only proves two of the examples in similar times as \vampire{}.

\section{Examples of Logic Programs with Pooling}
  The following four examples contain pools which \anthem{} can handle thanks to the extension of~$\tau^*$ defined in Section~\ref{sec:translation}.
  The results of the theorem provers on these examples can be found in Table~\ref{tab:pooling}

  The first example are Programs~\ref{lst:colour_atoms}~and~\ref{lst:colour_pool} from Chapter~\ref{chap:introduction}, which define the available colours in a graph colouring problem.

  \vspace{\extop}\noindent
  \begin{tabular}{p{\excol}  p{\excol}}
    \textbf{Program 8.A} & \textbf{Program 8.B} \\
    \begin{tabularlstlisting}[language=input]
colour(r).
colour(g).
colour(b).
    \end{tabularlstlisting}
    &
    \begin{tabularlstlisting}[language=input]
colour(r;g;b).
    \end{tabularlstlisting}
  \end{tabular}
  All theorem provers manage to verify the strong equivalence of the programs very quickly.

  The next example shows how a pool can be used to replace an interval.

  \vspace{\extop}\noindent
  \begin{tabular}{p{\excol}  p{\excol}}
    \textbf{Program 9.A} & \textbf{Program 9.B} \\
    \begin{tabularlstlisting}[language=input]
p(1..3).
    \end{tabularlstlisting}
    &
    \begin{tabularlstlisting}[language=input]
p(1;2;3).
    \end{tabularlstlisting}
  \end{tabular}
  While this example seems quite similar to the previous example, both \princess{} and \vampire{} need significantly longer to verify the strong equivalence.
  Using the parallelisation of \vampire{} can again reduce the time.
  \cvc{} also verifies this example quickly, while \zipperposition{} cannot find a proof in the time-limit.

  The following example shows how two rules can be combined into a single rule using a pool in the body of the rule.

  \vspace{\extop}\noindent
  \begin{tabular}{p{\excol}  p{\excol}}
    \textbf{Program 10.A} & \textbf{Program 10.B} \\
    \begin{tabularlstlisting}[language=input]
p(X) :- q(X), X=a.
P(X) :- q(X), X=1..5.
    \end{tabularlstlisting}
    &
    \begin{tabularlstlisting}[language=input]
p(X) :- q(X), X=(a;1..5).
    \end{tabularlstlisting}
  \end{tabular}
  On this example, \princess{} and \zipperposition{} cannot find a proof within the time-limit.
  \cvc{} and \vampire{} manage to verify the strong equivalence in less than a second.

  The final example shows how to write Rule~\ref{eq:pool_head} without using a pool.

  \vspace{\extop}\noindent
  \begin{tabular}{p{\excol}  p{\excol}}
    \textbf{Program 11.A} & \textbf{Program 11.B} \\
    \begin{tabularlstlisting}[language=input]
p(X) :- q(X,Y).
p(Y) :- q(X,Y).
    \end{tabularlstlisting}
    &
    \begin{tabularlstlisting}[language=input]
p(X;Y) :- q(X,Y).
    \end{tabularlstlisting}
  \end{tabular}
  This example is also verified very quickly by \cvc{} and \vampire{}.
  Although a bit slower \princess{} also manages to verify the example, while \zipperposition{} again runs into a timeout.

  \begin{table}
    \begin{adjustbox}{max width=\textwidth}
      \begin{tabular}{| l *{6}{| c} |}
        \hline
         & \cvc & \princess & \vampire & \vampirep & \zipperposition \\
        \hline
        Example 8 & 0.014 s & 1.075 s & 0.008 s & 0.011 s & 0.075 s \\
        \hline
        Example 9 & 0.045 s & 20.226 s & 23.85 s & 4.483 s & --- \\
        \hline
        Example 10 & 0.031 s & --- & 0.121 s & 0.185 s & --- \\
        \hline
        Example 11 & 0.022 s & 0.825 s & 0.01 s & 0.011 s & --- \\
        \hline
      \end{tabular}
    \end{adjustbox}
    \caption{Results on examples of logic programs with pooling}
    \label{tab:pooling}
  \end{table}

\section{Examples of Logic Programs with Negation}
  The next seven examples are programs with negation which \anthem{} can translate to classical first-order logic by making use of~$\sigma^*$ (Section~\ref{sec:semantics_map}).
  The results of the theorem provers on the examples which are strongly equivalent can be found in Table~\ref{tab:negation}.

  The first example replaces a program stating that \lstinline{p} is true when \lstinline{q} or the negation of \lstinline{q} holds by a fact.

  \vspace{\extop}\noindent
  \begin{tabular}{p{\excol}  p{\excol}}
    \textbf{Program 12.A} & \textbf{Program 12.B} \\
    \begin{tabularlstlisting}[language=input]
p :- q.
p :- not q.
    \end{tabularlstlisting}
    &
    \begin{tabularlstlisting}[language=input]
p.
    \end{tabularlstlisting}
  \end{tabular}
  This example corresponds to the example of uniformly equivalent programs, Programs~\ref{ex:uniform_a} and~\ref{ex:uniform_b}, from Chapter~\ref{chap:introduction}.
  The two programs are not strongly equivalent, as by adding the rule \lstinline{q :- p} Program 12.A does not have an answer set while Program 12.B has the answer set
  $\{p, q\}$.
  Therefore, all theorem provers are unable to verify the strong equivalence.

  The next example shows how the previous two programs can be made strongly equivalent: by adding the constraint \lstinline{:- q}.

  \vspace{\extop}\noindent
  \begin{tabular}{p{\excol}  p{\excol}}
    \textbf{Program 13.A} & \textbf{Program 13.B} \\
    \begin{tabularlstlisting}[language=input]
p :- q.
p :- not q.
:- q.
    \end{tabularlstlisting}
    &
    \begin{tabularlstlisting}[language=input]
p.
:- q.
    \end{tabularlstlisting}
  \end{tabular}
  All theorem provers verify that the two programs are strongly equivalent in under a second with \cvc{} and \vampire{} being the quickest.

  The following example shows that the constraint \lstinline{:- not p, q} is redundant in the presence of the rule \lstinline{p :- q}.

  \vspace{\extop}\noindent
  \begin{tabular}{p{\excol}  p{\excol}}
    \textbf{Program 14.A} & \textbf{Program 14.B} \\
    \begin{tabularlstlisting}[language=input]
p :- q.
:- not p, q.
    \end{tabularlstlisting}
    &
    \begin{tabularlstlisting}[language=input]
p :- q.
    \end{tabularlstlisting}
  \end{tabular}
  This example corresponds to the example of strongly equivalent programs, Programs~\ref{ex:strong_eq2_a} and~\ref{ex:strong_eq2_b}, from Chapter~\ref{chap:introduction}.
  Again, \cvc{} and \vampire{} verify the example in a few milliseconds, \zipperposition{} is slightly slower and \princess{} while being the slowest, is still under a second.

  Example~15 is similar to the previous example; however, the programs are more complex as the body atom is now negated, and the rules make use of simple arithmetic expressions.

  \vspace{\extop}\noindent
  \begin{tabular}{p{\excol}  p{\excol}}
    \textbf{Program 15.A} & \textbf{Program 15.B} \\
    \begin{tabularlstlisting}[language=input]
p(2*X) :- not q(X).
:- not p(X+X), not q(X).
    \end{tabularlstlisting}
    &
    \begin{tabularlstlisting}[language=input]
p(2*X) :- not q(X).
    \end{tabularlstlisting}
  \end{tabular}
  The higher complexity of the programs is reflected in the results of the theorem provers.
  \princess{}, \zipperposition{}, and \vampire{} (even using parallelisation) cannot prove the strong equivalence of the programs in 300~s.
  Only \cvc{} verifies the example and does so in just 24~ms.

  The next example shows that the rule \lstinline{p :- not p} (which on its own does not produce an answer set) can be removed from a program stating that~\lstinline{p} is true when~\lstinline{q}~or~\lstinline{not q} holds.

  \vspace{\extop}\noindent
  \begin{tabular}{p{\excol}  p{\excol}}
    \textbf{Program 16.A} & \textbf{Program 16.B} \\
    \begin{tabularlstlisting}[language=input]
p :- q.
p :- not q.
p :- not p.
    \end{tabularlstlisting}
    &
    \begin{tabularlstlisting}[language=input]
p :- q.
p :- not q.
    \end{tabularlstlisting}
  \end{tabular}
  In this example, all theorem provers are successful again, verifying the strong equivalence in less than one second.

  The following example shows that the single rule Program~17.A (which has an unsatisfiable body) is strongly equivalent to the empty program represented by the trivial rule \lstinline{p :- p}.

  \vspace{\extop}\noindent
  \begin{tabular}{p{\excol}  p{\excol}}
    \textbf{Program 17.A} & \textbf{Program 17.B} \\
    \begin{tabularlstlisting}[language=input]
p :- q, not q.
    \end{tabularlstlisting}
    &
    \begin{tabularlstlisting}[language=input]
p :- p.
    \end{tabularlstlisting}
  \end{tabular}
  To verify this \princess{} needs a bit more than half a second while the other provers verify the example in 10~ms or less.

  The final example of this section contains two slightly longer programs.
  The final rule of the first program, \lstinline{s :- not r}, can be safely removed from the program as at least one of~\lstinline{p}~and~\lstinline{q} always holds (because of the first two rules) and so the body of either rule three or four is satisfied.

  \vspace{\extop}\noindent
  \begin{tabular}{p{\excol}  p{\excol}}
    \textbf{Program 18.A} & \textbf{Program 18.B} \\
    \begin{tabularlstlisting}[language=input]
p :- not q.
q :- not p.
r :- p, q.
s :- p.
s :- q.
s :- not r.
    \end{tabularlstlisting}
    &
    \begin{tabularlstlisting}[language=input]
p :- not q.
q :- not p.
r :- p, q.
s :- p.
s :- q.
    \end{tabularlstlisting}
  \end{tabular}
  The strong equivalence of the programs is verified by all theorem provers.
  \cvc{} and \vampirep{} both verify this in around 10~ms.
  \princess{} and \zipperposition{} verify the example in under one second.
  Without the parallelisation \vampire{} takes the longest on this example with slightly over two seconds.

  \begin{table}
    \begin{adjustbox}{max width=\textwidth}
      \begin{tabular}{| l *{6}{| c} |}
        \hline
         & \cvc & \princess & \vampire & \vampirep & \zipperposition \\
        \hline
        Example 13 & 0.010 s & 0.616 s & 0.004 s & 0.003 s & 0.183 s \\
        \hline
        Example 14 & 0.009 s & 0.753 s & 0.006 s & 0.005 s & 0.086 s \\
        \hline
        Example 15 & 0.024 s & --- & --- & --- & --- \\
        \hline
        Example 16 & 0.011 s & 0.810 s & 0.042 s & 0.005 s & 0.135 s \\
        \hline
        Example 17 & 0.010 s & 0.595 s & 0.004 s & 0.005 s & 0.009 s \\
        \hline
        Example 18 & 0.011 s & 0.818 s & 2.132 s & 0.009 s & 0.483 s \\
        \hline
      \end{tabular}
    \end{adjustbox}
    \caption{Results on examples of logic programs with negation}
    \label{tab:negation}
  \end{table}

\section{Examples of Logic Programs with Simple Choice Rules}
  Finally, the examples in this section are logic programs containing simple choice rules, i.e. with just a single atom, and negation.
  The results of the theorem provers on the strongly equivalent examples are given in Table~\ref{tab:choices}.

  The first example attempts to rewrite a logic program, stating that exactly one of~\lstinline{p}~and~\lstinline{q} has to be true, without using choice rules.

  \vspace{\extop}\noindent
  \begin{tabular}{p{\excol}  p{\excol}}
    \textbf{Program 19.A} & \textbf{Program 19.B} \\
    \begin{tabularlstlisting}[language=input]
{ p }.
{ q }.
:- p, q.
:- not p, not q.
    \end{tabularlstlisting}
    &
    \begin{tabularlstlisting}[language=input]
p :- not q.
q :- not p.
    \end{tabularlstlisting}
  \end{tabular}
  However, the two programs are not strongly equivalent as adding both~\lstinline{p}~and~\lstinline{q} as facts results in a valid answer set $\{p, q\}$ for Program~19.B while Program~19.A does not have an answer set with these facts.
  Consequently, none of the theorem provers manages to verify the strong equivalence.

  By adding the constraint \lstinline{:- p, q} to the second program, the two programs can be made strongly equivalent as shown in the following example.

  \vspace{\extop}\noindent
  \begin{tabular}{p{\excol}  p{\excol}}
    \textbf{Program 20.A} & \textbf{Program 20.B} \\
    \begin{tabularlstlisting}[language=input]
{ p }.
{ q }.
:- p, q.
:- not p, not q.
    \end{tabularlstlisting}
    &
    \begin{tabularlstlisting}[language=input]
p :- not q.
q :- not p.
:- p, q.
    \end{tabularlstlisting}
  \end{tabular}
  The strong equivalence of the two programs can successfully be verified by all theorem provers.
  \vampirep{} performs best on this example with both \cvc{} and \vampire{} being slightly slower.
  \princess{} and \zipperposition{} both take longer to verify the example but still do so in under one second.

  The following example consists of Program~\ref{eq:choice_negation} and the program obtained from replacing the constraint in Program~\ref{eq:choice_negation} by a simple rule.

  \vspace{\extop}\noindent
  \begin{tabular}{p{\excol}  p{\excol}}
    \textbf{Program 21.A} & \textbf{Program 21.B} \\
    \begin{tabularlstlisting}[language=input]
{ p }.
:- not p, q.
    \end{tabularlstlisting}
    &
    \begin{tabularlstlisting}[language=input]
{ p }.
p :- q.
    \end{tabularlstlisting}
  \end{tabular}
  Both configurations of \vampire{} prove the strong equivalence in a few milliseconds with \cvc{} being just slightly slower.
  \zipperposition{} and \princess{} again need a bit more time (with \princess{} being the slowest) but still verify the example in under a second.

  Next up are Programs~\ref{lst:transitive_rule}~and~\ref{lst:transitive_constraint} from Chapter~\ref{chap:introduction}, which define a transitive relation~\lstinline{p} for the set of values for which~\lstinline{q(X)} holds.

  \vspace{\extop}\noindent
  \begin{tabular}{p{\excol}  p{\excol}}
    \textbf{Program 22.A} & \textbf{Program 22.B} \\
    \begin{tabularlstlisting}[language=input]
{ p(X,Y) } :- q(X), q(Y).
p(X,Y) :- p(X,Z), p(Z,Y),
          q(X), q(Y), q(Z).
    \end{tabularlstlisting}
    &
    \begin{tabularlstlisting}[language=input]
{ p(X,Y) } :- q(X), q(Y).
:- p(X,Z), p(Z,Y), not p(X,Y),
   q(X), q(Y), q(Z).
    \end{tabularlstlisting}
  \end{tabular}
  \zipperposition{} is the only prover that does not manage to verify this example in under 300~s.
  \vampire{} takes several seconds but can be made much faster using parallelisation which brings the time to the same level as \cvc{}.
  \princess{} manages to find the proof faster than \vampire{} but cannot match \cvc{} or \vampirep{}.

  Program 23.A gives a complex way to express the fact~\lstinline{p} by using a choice over~\lstinline{p} as well as the rule \lstinline{p :- not p}.

  \vspace{\extop}\noindent
  \begin{tabular}{p{\excol}  p{\excol}}
    \textbf{Program 23.A} & \textbf{Program 23.B} \\
    \begin{tabularlstlisting}[language=input]
{ p }.
p :- not p.
    \end{tabularlstlisting}
    &
    \begin{tabularlstlisting}[language=input]
p.
    \end{tabularlstlisting}
  \end{tabular}
  All theorem prover manage to verify that Program~23.A is indeed strongly equivalent to the fact~\lstinline{p}.
  Again, \vampire{} (both configurations) does so in the least amount of time closely followed by \cvc{}, followed by \zipperposition{} and \princess{} again the slowest but still under one second.

  Finally, Program~24.B attempts to simplify Program~24.A by removing the constraint \lstinline{:- p(X), not q(X)} and instead adding the condition~\lstinline{q(X)} to the body of the choice rule.

  \vspace{\extop}\noindent
  \begin{tabular}{p{\excol}  p{\excol}}
    \textbf{Program 24.A} & \textbf{Program 24.B} \\
    \begin{tabularlstlisting}[language=input]
{ p(X) }.
:- p(X), not q(X).

    \end{tabularlstlisting}
    &
    \begin{tabularlstlisting}[language=input]
{ p(X) } :- q(X).
    \end{tabularlstlisting}
  \end{tabular}
  This is, however, not a strongly equivalent transformation, as simply adding the fact~\lstinline{p(1)} results in Program~24.A being unsatisfiable while Program~24.B has an answer set consisting of just~\lstinline{p(1)}.
  Thus, none of the theorem provers verifies the strong equivalence of the two programs.

  \begin{table}
    \begin{adjustbox}{max width=\textwidth}
      \begin{tabular}{| l *{6}{| c} |}
        \hline
         & \cvc & \princess & \vampire & \vampirep & \zipperposition \\
        \hline
        Example 20 & 0.017 s & 0.758 s & 0.028 s & 0.006 s & 0.275 s \\
        \hline
        Example 21 & 0.011 s & 0.707 s & 0.004 s & 0.005 s & 0.139 s \\
        \hline
        Example 22 & 0.022 s & 1.092 s & 5.21 s & 0.025 s & --- \\
        \hline
        Example 23 & 0.019 s & 0.596 s & 0.003 s & 0.005 s & 0.144 s \\
        \hline
      \end{tabular}
    \end{adjustbox}
    \caption{Results on examples of logic programs with simple choice rules}
    \label{tab:choices}
  \end{table}

\section{Summary}
  Overall the comparison does not show any of the theorem provers vastly outperforming the others.

  \vampire{} managed to verify the most examples, just running into a timeout for one example.
  \cvc{} only failed to do so on four of the examples with positive programs.
  \princess{} failed on seven examples, and \zipperposition{} failed on nine examples.

  In most examples, \cvc{} has the shortest run time.
  Using the parallelisation of \vampire{} brings significant improvements in most examples, often bringing the time onto a similar level as \cvc{} or even making \vampire{} faster.
  Both \princess{} and \zipperposition{} are most of the time quite a bit slower than \cvc{} and \vampire{} (with \princess{} being the slower one), but they still manage to verify most examples in one second or less.

  In conclusion, \cvc{} and \vampire{} seem to be the best fit for the kind of problems generated by \anthem{}.
  While both \princess{} and \zipperposition{} do not perform much worse, they are not a better choice (than \cvc{} or \vampire{}) in any of the examples and run into timeouts more often compared to \cvc{} and \vampire{}.
  Therefore, using \cvc{} and \vampire{} in conjunction with \anthem{} seems to be the best option.

  \chapter{Conclusion}
\label{chap:conclusion}
  This chapter summarises the contributions of this thesis to extend the translation tool \anthem{} to enable the verification of strong equivalence of non-positive logic programs.
  The contributions are compared to related work.
  Afterwards, possible areas of future work are briefly discussed.

\section{Contributions and Related Work}
  This section summarises the four main contributions of this thesis and compares them to related work.

  First, the new transformation~$\sigma^*$ was introduced.
  A similar transformation was first introduced by~\cite{PearceEtAl2001EncodingsEquilibriumLogic}.
  While parts of the formalisation in this thesis correspond to the one in~\cite{PearceEtAl2001EncodingsEquilibriumLogic} (i.e. Lemma~\ref{lem:satisfiability} in this thesis is similar to~\cite[Lemma~2]{PearceEtAl2001EncodingsEquilibriumLogic}), the transformation was used differently by~\cite{PearceEtAl2001EncodingsEquilibriumLogic}.
  It was only applied to expressions (i.e. formulas only consisting of atoms and the connectives $\land, \lor,$ and $\neg$) in order to express equilibrium models~\cite{Pearce1997NewLogicalCharacterisation} in classical logic.
  In this thesis, $\sigma^*$~is used to express the equivalence of two formulas in the logic of here-and-there in classical logic, which is formalised in Theorem~\ref{thm:ht_equivalence}.

  Second, the translation~$\tau^*$ was extended to cover an input language containing pools.
  Theorem~\ref{thm:strong_eq_classical} formalises how the extended~$\tau^*$ can be combined with~$\sigma^*$ to express the strong equivalence of two logic programs in classical logic.

  Third, a new version of \anthem{} implements the extended translation~$\tau^*$ as well as the transformation~$\sigma^*$ enabling the verification of strong equivalence of non-positive programs by using an automated theorem prover for classical logic.
  A similar system was developed by~\cite{ChenEtAl2005SELPSystemStudying}, which utilises a transformation similar to~$\sigma^*$.
  However, the supported input language is much more limited compared to the input language of \anthem{}.
  While disjunctions are supported, choices, comparisons, arithmetic expressions, intervals, and pools are not supported.
  Another difference is how the verification is done.
  The system of~\cite{ChenEtAl2005SELPSystemStudying} uses a SAT solver and can generate counter-examples for programs that are not strongly equivalent.

  Fourth, options for automated theorem provers which are compatible with \anthem{} are investigated.
  Four automated theorem provers (\cvc{}, \princess{}, \vampire{}, and \zipperposition{}) are tested on several examples.
  The examples include positive programs, programs with pools, which \anthem{} can translated as a result of the extensions of~$\tau^*$, as well as programs with negation and simple choices, which \anthem{} translates to classical logic by using~$\sigma^*$.
  \cvc{} and \vampire{} emerged as the best options to use with \anthem{}.
  \princess{} and \zipperposition{} both only manage to verify fewer examples compared to \cvc{} and \vampire{}.

\section{Future Work}
  Future work on \anthem{} is mainly concerned with extending the input language.
  Desirable language features to support include an extended version of choices and aggregates, as both are essential features in most applications of Answer Set Programming.
  Other possible features include disjunctions in rule heads, classical negation, and conditionals in rule bodies.

  For more complex logic programs it will most likely be necessary to assist the theorem provers in the verification process.
  A method to do so has already been implemented for the version of \anthem{} used to verify the correctness of logic programs~\cite{FandinnoEtAl2020VerifyingTightLogic}.
  Integrating this method into the verification of strong equivalence of logic programs is another direction of future work.

	\backmatter
  \bibliographystyle{apalike}
  \bibliography{references}
  \chapter*{Declaration of Authorship}
  I hereby declare that I have independently written the present thesis and have not used any sources other than those I have specified.
  \bigskip

  \noindent Hiermit erkläre ich, dass ich die vorliegende Arbeit selbstständig verfasst und keine anderen als die von mir angegebenen Quellen genutzt habe.
  \vspace{2 cm}

  \noindent\textbf{\theauthor}

  \noindent \theplace, \thedate

\end{document}